\documentclass{amsart}
\usepackage{epsfig}

\usepackage{url}
\usepackage{algorithm}
\usepackage{algorithmic}

\usepackage{multirow}
\usepackage{amssymb}
\usepackage{graphicx}
\usepackage{epsfig}
\usepackage{amsfonts,amssymb,amscd,amsmath,euscript,enumerate,verbatim,calc}

\usepackage{amsfonts,amssymb,amscd,amsmath,euscript,enumerate,verbatim,calc}
\theoremstyle{plain}

\newtheorem{theorem}{Theorem}[section]
\newtheorem{lemma}[theorem]{Lemma}
\newtheorem{proposition}[theorem]{Proposition}

\newtheorem{corollary}[theorem]{Corollary}
\theoremstyle{definition}
\newtheorem{definition}[theorem]{Definition}

\newtheorem{remark}[theorem]{Remark}
\newtheorem{example}[theorem]{Example}

\def\Fq{{\mathbb F}_q}
\def\Fqm{{\mathbb F}_{q^m}}

\newcommand{\MRMMP}{\operatorname{MRMMP}}
\newcommand{\MRMM}{\operatorname{MRMM}}

\newcommand{\GL}{\operatorname{GL}}

\newcommand{\diag}{\operatorname{diag}}

\newcommand{\GF}{\operatorname{GF}}
\begin{document}

\title[The multiple-recursive matrix method]
{A note on the multiple-recursive matrix method for generating pseudorandom vectors}
\author{Susil Kumar Bishoi}
\address{Center for Artificial Intelligence and Robotics, Defence Research and Development Organisation, CV Raman Nagar, Bengaluru 560093, India}
\email{skbishoi@gmail.com}

\author{Himanshu Kumar Haran}
\address{Center for Artificial Intelligence and Robotics, Defence Research and Development Organisation, CV Raman Nagar, Bengaluru 560093, India}
\email{himanshu.haran@gmail.com}

\author{Sartaj Ul Hasan}
\address{Scientific Analysis Group, Defence Research and Development
Organisation, Metcalfe House, Delhi 110054, India}
\email{sartajulhasan@gmail.com}

\keywords{Linear feedback shift register; multiple-recursive matrix method; 
linear complexity; 
Langford arrangement.}

\subjclass[2010]{94A55, 94A60, 15B33, 12E20 and 12E05.}

\begin{abstract}
The multiple-recursive matrix method for generating pseudorandom vectors was introduced by Niederreiter (Linear Algebra Appl. 192 (1993), 301-328). We propose an algorithm for finding an efficient primitive multiple-recursive matrix method. Moreover, for improving the linear complexity, we introduce a tweak on the contents of the primitive multiple-recursive matrix method. 
\end{abstract}
\date{\today}
\maketitle

\section{Introduction}\label{intro}
In most of the modern stream ciphers, we generally use linear feedback shift registers (LFSRs) as basic building blocks that produce only one new bit per step. Such ciphers are often referred to as bit-oriented. 
Bit-oriented ciphers not only have large period and good statistical properties, but also have low cost of implementation in hardware and thus are quite useful in applications like wireless communications. However, 
in many situations such as high speed link encryption, an efficient software encryption is required and bit-oriented ciphers do not provide adequate efficiency. 

The question arises: How to design feedback shift registers (FSRs) that output a word instead of a bit per clock? A very natural and obvious way is to consider FSRs over extension fields, but then field multiplication 
being an expensive operation, it would not really make our life easy in terms of software efficiency. The other way is to exploit word operations -- logic operations and arithmetic operations -- of modern computer processors 
in designing FSRs so as to enhance their efficiency in software implementation. In fact, Preneel in \cite{BP} poses the question of whether one can design fast and secure FSRs with the help of the word operations of 
modern processors and the techniques of parallelism. 

Interestingly, a solution to Preneel's problem was already available in the literature even before it was formally stated and was given by Niederreiter in a series of papers \cite{N1,N2,N3,N4} by introducing 
the multiple-recursive matrix method (MRMM) for generating pseudorandom vectors. This method involves matrix multiplication which is again an expensive operation as far as software efficiency is concerned. Zeng {\it et al.} \cite{Zeng} 
resolve this problem by imposing restriction on the choice of matrices used in the multiple-recursive matrix method. In fact, Zeng {\it et al.} \cite{Zeng} introduce the notion of word-oriented $\sigma$-LFSR and it turns out that the 
seemingly new notion of $\sigma$-LFSR is essentially equivalent to the multiple-recursive matrix method for generating pseudorandom vectors. They also gave a conjectural formula for the number of primitive $\sigma$-LFSRs. 
This conjecture has been proved in the affirmative and the reader is referred to \cite{CT,GSM,GR,GR1,KP1,KP2} for more details. Throughout this paper, we shall use the acronym MRMM instead of $\sigma$-LFSR.

It may be noted that Tsaban and Vishne \cite{TV} also addressed the problem of Preneel by introducing the notion of transformation shift registers (TSR). It turns out that TSR is a special case of the multiple-recursive 
matrix method. One may refer to \cite{CHPW,HPW,JY,Ram} for some recent progress concerning TSRs.

By choosing matrices at random from a special set of matrices that are compatible with word operations of modern processors, a search algorithm for finding some efficient primitive MRMM was proposed in \cite[Algorithm 1]{Zeng}. 
We would like to reiterate here that by efficient, we mean that the matrices used in  MRMM can simply be replaced by word operations while computing the feedback. 
For all practical purposes where software efficiency is of paramount importance, we need 
an algorithm for explicitly constructing efficient primitive MRMM. 
It may be noted that a method for constructing primitive MRMM can be gleaned from the proof of \cite[Theorem 6.1]{GSM}, but it neither constructs all the primitive MRMM nor does it produce an efficient primitive MRMM. Lachaud \cite{Gilles} and Krishnaswamy {\it et al.} \cite{KP2} also proposed nice methods for constructing all of the primitive MRMM. 
In this paper, however, we are not really focusing on constructing all of the primitive MRMM. In fact, we are interested in constructing only some efficient ones so that the problem of software efficiency in various applications is resolved. Very recently, M. A. Goltvanitsa has drawn our attention to \cite[Section~2.2]{GNZ} which also discusses a similar construction for non-linearized skew MP-polynomials. However, it appears that our techniques are completely different from those discussed in \cite{GNZ}. 

 As we know, linear complexity plays a crucial role in determining the security of the keystream generated by FSRs. 
 In order to enhance linear complexity of sequences generated by the multiple-recursive matrix method, 
 one might consider employing some nonlinear functions on its contents. In fact, in \cite{HPW}, a nonlinear scheme based on Langford arrangement was employed on sequences generated by primitive TSRs. We replicate a similar, yet slightly different tweak, for the sequences generated by primitive MRMM along the similar lines. Since this ``little tweak'' has not yet been reported in the literature, we thought of including it in the form of the tweaked primitive multiple-recursive matrix method for the sake of completeness.

The paper is organized as follows. In Section \ref{secmrmm}, we recall some definitions and results concerning the multiple-recursive matrix method that are needed in this work. We develop some mathematical theory for constructing 
efficient MRMM in Section \ref{construction}. 
We propose an algorithm for finding efficient primitive MRMM in Section \ref{algorithm}. 
In Section \ref{implementation}, we discuss implementation issues of MRMM obtained through our 
algorithm. 
Finally, in Section \ref{nmrmm}, we discuss the tweaked primitive MRMM based on Langford arrangement.
\section {The Multiple-Recursive Matrix Method} \label{secmrmm}
We denote by $\Fq$ the finite field with $q$ elements,
where $q$ is a prime power and by $\Fq[X]$ the ring of polynomials in one variable $X$
with coefficients in $\Fq$. 
Also we denote by
$M_d(\Fq)$ the set of all $d\times d$ matrices with entries
in $\Fq$. We now recall some definitions and
results from \cite{GSM,N2} concerning the multiple-recursive matrix method.

In what follows, we fix positive
integers $m$ and $n$, and a vector space basis $\{\alpha_0,
\dots, \alpha_{m-1}\}$ of ${\mathbb F}_{q^m}$ over
$\mathbb F_q$. Given any $ s\in {\mathbb F}_{q^m}$, there are
unique $s_0, \dots, s_{m-1} \in {\mathbb F}_{q}$ such that
$s= s_0 \alpha_0 + \cdots + s_{m-1}\alpha_{m-1}$,
and we shall denote the corresponding co-ordinate vector
$(s_0, \dots, s_{m-1})$ of $s$ by $\mathbf{s}$. Evidently,
the association $s\longmapsto \mathbf{s}$ gives a vector
space isomorphism of $\mathbb F_{q^m}$ onto $\mathbb F_q^m$.
Elements of $\mathbb F_q^m$ may be thought of as column vectors
and so $\, C\mathbf{s}$ is a well-defined element of
$\mathbb F_q^m$ for any $\mathbf{s} \in \mathbb F_q^m$ and
$C\in M_m(\Fq)$.

\begin{definition}
\label{def:sigmalfsr}
Let $C_0, C_1, \dots, C_{n-1} \in M_m(\mathbb F_q)$. Given any 
$n$-tuple $(\mathbf{s}_0, \dots, \mathbf{s}_{n-1})$  of elements 
of $\mathbb F_{q^m}$, let $(\mathbf{s}_i)_{i=0}^{\infty}$ 
denote the infinite sequence of elements of ${\mathbb F}_{q^m}$  
determined by the following linear recurrence relation:
\begin{eqnarray}
{\mathbf{s}}_{i+n}=C_0{\mathbf{s}}_i+C_1{\mathbf{s}}_{i+1}+\cdots +C_{n-1}{\mathbf{s}}_{i+n-1} \quad i=0,1,\dots. \label{sigmalfsr} 
\end{eqnarray}
The system \eqref {sigmalfsr} is called the \emph{multiple-recursive matrix method (MRMM)} of order $n$ over 
$\mathbb F_{q^m}$, while the sequence $(\mathbf{s}_i)_{i=0}^{\infty}$  
is referred to as the \emph{sequence generated by the 
MRMM} \eqref{sigmalfsr}. The $n$-tuple 
$(\mathbf{s}_0,\mathbf{s}_1, \ldots, \mathbf{s}_{n-1})$ is 
called \emph{initial state} of the MRMM \eqref{sigmalfsr} 
and the polynomial $I_mX^n -C_{n-1}X^{n-1}- \cdots  -C_1X-C_0$ with 
matrix coefficients is called the \emph{matrix polynomial} of 
the MRMM \eqref{sigmalfsr}. The sequence 
$(\mathbf{s}_i)_{i=0}^{\infty}$  is \emph{ultimately periodic} 
if there are integers $r, n_0$ with $r\ge 1$ and $n_0\geq0$ 
such that $\mathbf{s}_{j+r}=\mathbf{s}_j$ for all $j \geq n_0$. 
The least positive integer $r$ with this property is called the 
\emph{period} of $(\mathbf{s}_i)_{i=0}^{\infty}$  and the 
corresponding least nonnegative integer $n_0$ is called the 
\emph{preperiod} of $(\mathbf{s}_i)_{i=0}^{\infty}$. The sequence 
$(\mathbf{s}_i)_{i=0}^{\infty}$ is said to be \emph{periodic} if 
its preperiod is $0$.
\end{definition}
The following proposition \cite[Proposition 4.2]{GSM} gives some basic facts about MRMM.
\begin{proposition} \label{tsrulti}
For the sequence $(\mathbf{s}_i)_{i=0}^{\infty}$ generated by
the MRMM $(\ref{sigmalfsr})$ of order $n$ over $\mathbb F_{q^m}$, we have
\begin{enumerate}
\item[{\rm (i)}] $(\mathbf{s}_i)_{i=0}^{\infty}$ is ultimately
periodic, and its period is no more than $q^{mn}-1$;
\item[{\rm (ii)}] if $C_0$ is nonsingular, then
$(\mathbf{s}_i)_{i=0}^{\infty}$ is periodic; conversely, if
$(\mathbf{s}_i)_{i=0}^{\infty}$ is periodic whenever the
initial state is of the form $(b, 0, \dots , 0)$, where
$b\in \mathbb F_{q^m}$ with $b\ne 0$, then $C_0$ is nonsingular.
\end{enumerate}
\end{proposition}

An MRMM of order $n$ over $\mathbb F_{q^m}$ is \emph{primitive} if
for any choice of nonzero initial state, the sequence generated
by that MRMM is periodic of period  $q^{mn}-1$.

In view of Proposition \ref{tsrulti} if $I_mX^n -C_{n-1}X^{n-1}- \cdots  -C_1X-C_0 \in M_m \left(\Fq\right)[X]$ is the matrix polynomial of a primitive MRMM, then the matrix $C_0$ is necessarily nonsingular.

Corresponding to a matrix polynomial $I_mX^n -C_{n-1}X^{n-1}-
\cdots-C_1X-C_0 \in M_m(\Fq)[X]$, we can associate a $(m,n)$-block
companion matrix $T \in M_{mn}(\Fq)$ of the following form
\begin{equation} \label{typeP}
T =
\begin {pmatrix}
\mathbf{0} & \mathbf{0} & \mathbf{0} & . & . & \mathbf{0} & \mathbf{0} & C_0\\
I_m & \mathbf{0} & \mathbf{0} & . & . & \mathbf{0} & \mathbf{0} & C_1\\
. & . & . & . & . & . & . & .\\
. & . & . & . & . & . & . & .\\
\mathbf{0} & \mathbf{0} & \mathbf{0} & . & . & I_m & \mathbf{0} & C_{n-2}\\
\mathbf{0} & \mathbf{0} & \mathbf{0} & . & . & \mathbf{0} & I_m & C_{n-1}
\end {pmatrix},
\end{equation}
where $I_m$ denotes the $m\times m$ identity matrix over $\Fq$, while
$\mathbf{0}$ indicates the zero matrix in $M_m(\Fq)$. The set of all such
$(m,n)$-block companion matrices $T$ over $\Fq$ shall be denoted
by $\MRMM(m,n,q)$. 
Using a Laplace expansion or a suitable sequence of
elementary column operations, we conclude that if $T \in \MRMM(m,n,q)$
is given by \eqref{typeP}, then $\det T = \pm \det (C_0)$.
Consequently,
\begin{equation}
\label{nonsingP}
T \in \GL_{mn}(\Fq) \Longleftrightarrow C_0\in \GL_m(\Fq).
\end{equation}
where $\GL_m(\Fq)$ is the general linear group of all $m \times m$
nonsingular matrices over $\Fq$.

It may be noted that the block companion matrix \eqref{typeP} is the
state transition matrix for the MRMM \eqref{sigmalfsr}. Indeed, the $k$-th
state $\mathbf{S}_k:=\left(\mathbf{s}_{k}, \mathbf{s}_{k+1}, \dots,
\mathbf{s}_{k+n-1}\right) \in \Fqm^n$ of the MRMM \eqref{sigmalfsr} is
obtained from the initial state
$\mathbf{S}_0:=\left(\mathbf{s}_{0}, \mathbf{s}_{1}, \dots,
\mathbf{s}_{n-1}\right) \in \Fqm^n$ by $\mathbf{S}_k = \mathbf{S}_0 T^k$,
for any $k\ge 0$. We can identify MRMM \eqref{sigmalfsr} with block companion matrix \eqref{typeP}.

The following lemma \cite[Lemma 5.1]{GSM} reduces the calculation of an $mn\times mn$ determinant to an
$m\times m$ determinant. 
\begin{lemma}
\label{mntom}
Let $T \in \MRMM(m,n,q)$ be given as in \eqref{typeP} and also let
$M(X)\in M_m\left(\Fq[X]\right)$ be defined by
$M(X) := I_m X^n - C_{n-1}X^{n-1} - \cdots - C_1 X - C_0$.
Then the characteristic polynomial of $T$ is equal to
$\det \left(M(X)\right)$.
\end{lemma}
The following characterization of primitive MRMM can be easily extracted from the results given in \cite{GSM} (see also \cite[Theorem 4]{N1}).
\begin{proposition}
\label{primtsr2}
Let $C_0 \in \GL_m(\Fq)$. Then the following are equivalent:
\begin{enumerate}
\item[{\rm (i)}] an MRMM \eqref{sigmalfsr} of order $n$ over $\Fqm$ is 
primitive;

\item[{\rm (ii)}]$o(T) = q^{mn}-1$, where $o(T)$ denotes the multiplicative order of $T$ in $\GL_{mn}(\Fq)$ ;

\item[{\rm (iii)}] $\det \left(M(X)\right)$ is a primitive polynomial 
over $\Fq$ of degree $mn$,  where $M(X)$ is same as defined in Lemma \ref{mntom}. 
\end{enumerate} 
\end{proposition}
We recall a lemma \cite[Lemma 1]{N2} that enables us to determine the linear complexity of sequences generated by primitive MRMM.
\begin{lemma}
\label{CWLC}
Let 
$$ \mathbf{s}_i=\left(s_i^{(1)},\dots,s_{i}^{(m)}\right) 
     \in \Fq^m \simeq \Fqm \quad  i=0,1,\dots, $$
be an arbitrary recursive vector sequence and let $h(X) \in \Fq[X]$ 
be the characteristic polynomial of the matrix $T$ in 
$(\ref{typeP})$. Then for each $1\leq j \leq m$ the sequence 
$s_0^{(j)},s_1^{(j)}, \dots$ of the $j^{\rm th}$ coordinates is a linear 
recurring sequence in $\Fq$ with characteristic polynomial $h(X)$.
\end{lemma}
The following corollary trivially follows from Lemma \ref{CWLC} and gives the 
component-wise linear complexity of the sequences generated by primitive MRMM.
\begin{corollary}
\label{compLC}
Let 
$$ \mathbf{s}_i=\left(s_i^{(1)},\dots,s_{i}^{(m)}\right) 
     \in \Fq^m \simeq \Fqm \quad i=0,1,\dots,$$ 
be a sequence generated by a primitive MRMM of order $n$ over 
$\Fqm$. Then for each $1\leq j \leq m$, the linear complexity 
of the $j^{\rm th}$ coordinate sequence $s_0^{(j)},s_1^{(j)}, \dots$ 
over $\Fq$ is $mn$.
\end{corollary}
An alternative statement of the Corollary \ref{compLC} can be found in \cite[Theorem 3]{Zeng}. Moreover, in view of Corollary \ref{compLC}, it is clear that if a sequence over $\Fqm$ generated by a primitive MRMM of order $n$ is viewed as a sequence over $\Fq$, then its linear complexity 
is $m^2n$. In effect, a primitive MRMM of order $n$ over $\Fqm$ is same as $m$ parallel primitive LFSRs of order $mn$. 

\section{Construction of the Multiple-Recursive Matrix Method}\label{construction}
As alluded to in the introduction, Zeng {\it et al.} proposed a search algorithm \cite[Algorithm 1]{Zeng} for generating efficient primitive MRMM. Their algorithm begins by randomly choosing some matrices that 
are compatible with word operations and then testing the primitivity of a polynomial obtained by computing the determinant of a matrix using Lemma \ref{mntom}. 

We begin this section by defining the notion of generalized Horner's form corresponding to a given polynomial. We then use it to construct an efficient primitive MRMM. It may be noted that the idea of using Horner's form in the context of LFSR may not be common, but it has been used to make jumping efficient as can be seen in \cite{HML}. 
 \begin{definition}
 Let $f(X)=\displaystyle\sum_{i=0}^{d}a_iX^{i}$ be  a polynomial of degree $d$ over $\Fq$. For any given positive integer $n\leq d$, we can find integers $m$ and $r$ such that $d=mn+r$, where $0\leq r<n$. We express $f(X)$ in the following form
\begin{equation} \label{horner}
f(X)=\displaystyle f_0+X^n\left(f_1+ X^n\left(f_2+ \cdots +X^n\left(f_{m-1}+X^nf_m\right)\cdots\right)\right),
\end{equation}  
where, 
$$f_i(X)=\displaystyle \sum^{(i+1)n-1}_{k=in}a_kX^{k-in}~~\mbox{for}~~ i=0,1,\ldots ,(m-1)~~\mbox{and}~~ f_m(X)=\displaystyle \sum^{d}_{k=mn}a_kX^{k-mn}.$$ The representation of $f(X)$ in \eqref{horner} is referred to as $n$-Horner's 
form \footnote{The $1$-Horner's form is indeed the usual Horner's form of a polynomial used for computing the polynomial value with less number of multiplications.} of $f(X)$. 
\end{definition}
\begin{example}\label{exam1}
 Consider the polynomial $f(X)=1+X^2+X^5+X^7+X^{10}+X^{11}+X^{12}\in \mathbb F_2[X]$ of degree $12$. Here $d=12$. For $n=3$, we have $m=4$ and $r=0$. Then $3$-Horner's form of $f(X)$ is given by
$$f(X)=f_0+X^3(f_1+X^3 (f_2+X^3 (f_3+X^3f_4 ))),$$
where $f_0(X)=(1+X^2)$, $f_1(X)=X^2$, $f_2(X)=X$, $f_3(X)=(X+X^2 )$ and $f_4(X)=1$.
\end{example}
Corresponding to the $n$-Horner's form of a given polynomial of degree $mn$ over $\Fq$, we can associate an $m \times m$ matrix as defined below. This matrix would play a crucial role in 
the construction of efficient multiple-recursive matrix method of order $n$ over $\Fqm$ for generating pseudorandom vectors.
\begin{definition}
Let $m$ and $n$ be positive integers and let $f(X)=\displaystyle\sum_{i=0}^{mn}a_iX^{i}$ be  polynomial of degree $mn$ over $\Fq$. 
The $m\times m$ matrix
\begin{equation} \label{hmatrix}
\begin{pmatrix}
  X^n    &    0   &  0   &  \cdots &    0   & f_0 \\
  -1      &    X^n &  0   &  \cdots &    0   & f_1 \\
  0      &    -1   &  X^n &  \cdots &    0   & f_2 \\
  \vdots & \vdots & \vdots & \ddots & \vdots & \vdots   \\
  0      &    0   &  0   &   \cdots &   X^n  & f_{m-2} \\
  0      &    0   &  0   &   \cdots &    -1   & f_{m-1}+f_mX^n 
\end{pmatrix},
\end{equation}
corresponding to the $n$-Horner's form \eqref{horner} of $f(X)$ is referred to as $n$-Horner's matrix of $f(X)$ and denoted as $H_m(n,f)$.
\end{definition}
For each $j=0, 1, \dots, n-1$, let $C_j$ denotes the $m\times m$ matrix whose entries are coefficients of $X^j$ in the matrix $H_m(n,f)$. It is easy to see that the matrix $H_m(n,f)$ can be written as
 \begin{equation} \label{mpoly}
 H_m(n,f)=I_mX^n+C_{n-1}X^{n-1}+ \cdots  +C_1X+C_0,
 \end{equation}
 provided $f$ is monic, that is, $f_m=1$.
 It is clear from \eqref{mpoly} that we can associate the multiple-recursive matrix method of order $n$ over $\Fqm$ corresponding to these $m \times m$ matrices $C_0, C_1,\dots, C_{n-1}$.

It is interesting to note that the matrix $C_j (1\leq j \leq n-1)$ has the following form
$$ 
C_j  = 
\begin{pmatrix}
  0    &    0   &  0   &  \cdots &    0   & a_j\\
  0      &    0 &  0   &  \cdots &    0   & a_{n+j} \\
  0      &    0   &  0 &  \cdots &    0   & a_{2n+j} \\
  \vdots & \vdots & \vdots & \ddots & \vdots & \vdots   \\
  0      &    0   &  0   &   \cdots &   0  & a_{(m-2)n+j} \\
  0      &    0   &  0   &   \cdots &    0   & a_{(m-1)n+j} 
\end{pmatrix}
$$
whose first $(m-1)$ columns are zero. Moreover, the matrix $C_0$ has the following form
  $$ 
C_0  = 
\begin{pmatrix}
  0    &    0   &  0   &  \cdots &    0   & a_0\\
  -1      &    0 &  0   &  \cdots &    0   & a_n \\
  0      &    -1   &  0 &  \cdots &    0   & a_{2n} \\
  \vdots & \vdots & \vdots & \ddots & \vdots & \vdots   \\
  0      &    0   &  0   &   \cdots &   0  & a_{(m-2)n} \\
  0      &    0   &  0   &   \cdots &    -1   & a_{(m-1)n}
\end{pmatrix}.
$$
It is due to this special structure of these matrices that we are able to construct an efficient multiple-recursive matrix method. In Section \ref{implementation}, we shall see in greater detail why such a construction is fast and efficient. 

The following lemma gives the determinant of the matrix $H_m(n,f)$ and will be used in the sequel. 
\begin{lemma}\label{dethmatrix}
 Let $H_m(n,f)$ be $n$-Horner's matrix corresponding to the polynomial $f(X)$ of degree $mn$ over $\Fq$ as defined in \eqref{hmatrix}. Then $\det \left(H_m(n,f)\right)$ is equal to $f(X)$. 
\end{lemma}
\begin{proof}
Add $X^n$ times the $n^{\rm th}$ row to the $(n-1)^{\rm th}$ row of the matrix $H_m(n,f)$. 
This will remove the $X^n$ in the $(n-1)^{\rm th}$ row and it will not alter the determinant. Next, add $X^n$ times the new $(n-1)^{\rm th}$ row to the $(n-2)^{\rm th}$ row. Continue successively until
all of the ${X^n}'s$ on the main diagonal have been removed. The result is the matrix

$$\begin{pmatrix}
  0    &    0   &  0   &  \cdots &    0   & f_0+X^n\left(f_1+ X^n\left(f_2+ \cdots +X^n\left(f_{m-1}+X^nf_m\right)\cdots\right)\right) \\
  -1      &    0 &  0   &  \cdots &    0   & f_1+ X^n\left(f_2+ \cdots +X^n\left(f_{m-1}+X^nf_m\right)\cdots\right) \\
  0      &    -1   &  0 &  \cdots &    0   & f_2+ \cdots +X^n\left(f_{m-1}+X^nf_m\right) \\
  \vdots & \vdots & \vdots & \ddots & \vdots & \vdots   \\
  0      &    0   &  0   &   \cdots &   0  & f_{m-2}+X^n\left(f_{m-1}+X^nf_m\right) \\
  0      &    0   &  0   &   \cdots &    -1   & f_{m-1}+X^nf_m 
\end{pmatrix} $$
which has the same determinant as $H_m(n,f)$.  
We can clean up the last column by adding to it appropriate multiples of the other columns so as to obtain 
$$ 
\det\left(H_m(n,f) \right) = 
\det \begin{pmatrix}
  0    &    0   &  0   &  \cdots &    0   & f(X)\\
  -1      &    0 &  0   &  \cdots &    0   & 0 \\
  0      &    -1   &  0 &  \cdots &    0   & 0 \\
  \vdots & \vdots & \vdots & \ddots & \vdots & \vdots   \\
  0      &    0   &  0   &   \cdots &   0  & 0 \\
  0      &    0   &  0   &   \cdots &    -1   & 0 
\end{pmatrix}.
$$
Finally, we can slide the last column to the first by successive column interchanges. We need $(m-1)$ interchanges, and so the determinant changes by $(-1)^{(m-1)}$. Further, if we pull out the negative 
sign in each of the rows in all except the first row, then the determinant gets multiplied by $(-1)^{(m-1)}$. It follows that the determinant of $H_m(n,f)$ is 
$(-1)^{2(m-1)}$ times the determinant of the diagonal matrix $\diag\left(f(X), 1, \dots ,1\right)$ and this proves the lemma. 
\end{proof}
We can construct an efficient MRMM regardless of whether it's reducible, irreducible or primitive. However, in view of cryptographic applications, we shall focus only in the construction of efficient primitive MRMM. 

From Proposition \ref{primtsr2}, it is clear that an MRMM \eqref{sigmalfsr} is \emph{primitive} if the characteristic polynomial $\det\left(M(X)\right)$ of its transition matrix $T$ is primitive of 
degree $mn$ over $\Fq$. We shall denote by $\MRMMP(m,n,q)$, the set of all those block companion matrices in $\MRMM(m,n,q)$ whose characteristic polynomial is primitive and by $\EuScript P(d,q)$, 
the set of primitive polynomials
 in $\Fq[X]$ of degree $d$.
Then 
the \emph{characteristic map}
$$ \Psi: M_{mn}(\Fq) \longrightarrow \Fq[X]\quad \mbox{defined by}
        \quad\Psi(T):=\det(XI_{mn}-T),$$
if restricted to the set $\MRMMP(m,n,q)$ 
yields the following map
$$\Psi_P:\MRMMP(m,n,q)\longrightarrow \EuScript P(mn,q).$$
By using the structure of Horner's matrix, we prove the surjectivity of the map $\Psi_P$ in the following theorem. 
The proof of this theorem would enable us a way to construct efficient primitive MRMM. 
\begin{theorem} \label{sur}
 The  map $\Psi_P:\MRMMP(m,n,q)\longrightarrow \EuScript P(mn,q)$ is surjective.
\end{theorem}
\begin{proof}
 Let $f(X)=\displaystyle\sum_{i=0}^{mn}a_iX^{i} \in  \EuScript P(mn,q)$. Clearly, $f$ is a monic polynomial i.e. $a_{mn}=1$. Therefore, as in \eqref{mpoly}, the $n$-Horner's matrix $H_m(n,f)$ of $f(X)$ can be expressed in the following form
 \begin{equation} \label{mpoly2}
 H_m(n,f)=I_mX^n+C_{n-1}X^{n-1}+ \cdots  +C_1X+C_0,
  \end{equation}
where $C_i$ denotes the $m\times m$ matrix whose entries are coefficients of $X^i$ in the Horner's matrix $H_m(n,f)$. Let $\widetilde{T} \in \MRMM(m,n,q)$ denotes the block companion matrix corresponding to 
the matrix polynomial \eqref{mpoly2}. Then by Lemma~\ref{mntom} and Lemma \ref{dethmatrix}, it follows that
$$\Psi_P(\widetilde{T})=\det(XI_{mn}-\widetilde{T})=\det\left(H_m(n,f) \right)=f(X),$$as desired. 
\end{proof}
\begin{remark} \label{rem1}
By taking the proof of Theorem \ref{sur} a step further, a short and elementary proof of \cite[Theorem 6.1]{GSM} follows immediately. In fact, it is easy to see that $\det \left(C_0\right)=\pm a_0$. 
Since $f(X)$ is primitive, we have $a_0 \neq 0$ and hence $C_0 \in \GL_m(\Fq)$. Thus in view of \eqref{nonsingP}, $\widetilde{T}\in \GL_{mn}(\Fq)$. Moreover, since characteristic polynomial $f(X)$ of $\widetilde{T}$ is 
primitive, it follows from Proposition \ref{primtsr2} that $o(\widetilde{T}) =~q^{mn}-1$. 
\end{remark}
\begin{remark}
It may also be interesting to note that a short and elementary proof of \cite[Proposition~2.2]{GR} can be derived
by considering $f(X)$ to be irreducible and following the similar lines as in the proof of Theorem \ref{sur}. 
\end{remark}
\section{The Algorithm} 
\label{algorithm}


In this section, we present an algorithm to find an efficient primitive MRMM of order $n$ over $\mathbb F_{q^m}$. In view of the proof of Theorem \ref{sur}, we shall begin by first finding a primitive polynomial $f(X)$ of degree $mn$ over $\mathbb F_q$ so as to obtain a primitive MRMM of order $n$ over $\mathbb F_{q^m}$. 

It may be remarked that for checking primitivity of a polynomial of degree $mn$ over $\Fq$, one needs to know the distinct prime factors of $q^{mn}-1$ beforehand. The computational complexity of finding distinct prime factors of $q^{mn}-1$ is very large. In fact, the factors of $q^{mn}-1$ can not be computed in polynomial time in general. However, for smaller values of $q$ (note that in most of applications $q$ is 2), many thanks to the Cunningham project \cite{BLSTS,W}, the factorization of $q^{mn}-1$ is known for reasonable values of $mn$ that are needed in most of practical applications. Our algorithm is based on the assumption that the distinct prime factors of $q^{mn}-1$ are already known. 
All the sequential steps are described in Algorithm ~\ref{Algo1}.
\begin{algorithm}
 \caption{: Finding an efficient primitive MRMM} \label{Algo1}
     \begin{algorithmic}[1]
     \REQUIRE Positive integers $m$ and $n$, the prime power $q$, the distinct prime factors $p_1, p_2, \cdots, p_k$ of $q^{mn}-1$.
     \ENSURE An efficient primitive MRMM of order $n$ over $\mathbb F_{q^m}$.
     \STATE Choose at random a monic polynomial $f\in \Fq[X]$ of degree $mn$. 
     This is done by randomly selecting $mn$ elements $a_0, a_1, \cdots, a_{mn-1}$ in $\Fq$ with $a_0 \neq 0$ and taking $f(X)=X^{mn}+a_{mn-1}X^{mn-1}+ \dots +a_1 X+a_0$.\label{step1}
     \STATE Verify if $f$ is irreducible. If $f$ is not irreducible then go to Step \ref{step1}, otherwise go to Step \ref{step3}. \label{step2}
     \STATE Verify if $f$ is primitive. If $f$ is not primitive then go to Step \ref{step1}, otherwise go to Step \ref{step4}. The primitivity test is done as follows: Compute $\displaystyle h(X)=X^{(q^{mn}-1)/p_i}$ mod $f(X)$ for each $i$. If $h(X) \neq 1$ 
     for all $k$ distinct prime factors $p_i$ then $f$ is primitive. \label{step3}
    \STATE Express $f(X)$ in its $n$-Horner's form. \label{step4}
    \STATE Construct $n$-Horner's matrix $H_m(n,f)$ of $f(X)$.
    \STATE Express $H_m(n,f)$ in the form of matrix polynomial as described in \eqref{mpoly}\label{step}, i.e.,
    $$ H_m(n,f)=I_mX^n+C_{n-1}X^{n-1}+ \cdots  +C_1X+C_0,$$
    where $C_i$ denotes the $m\times m$ matrix whose entries are coefficients of $X^i$ in the Horner's matrix $H_m(n,f)$.
    \STATE Return $C_0, C_1, \cdots, C_{n-1}$. 
  \end{algorithmic}
\end{algorithm}

In Step~\ref{step2} of the algorithm, one may use Ben-Or's algorithm \cite{BO,GP} for irreducibility test, which is quite efficient in practice. It is pointed out in \cite{GP} that by using fast multiplication \cite{CK,S,SS}, the worst case complexity of Ben-Or's algorithm is 
$O(m^2n^2 \log mn \log\log mn \log mnq)$. 
As noted in \cite[Section 1]{GVPS}, in polynomial basis representation of $\mathbb F_{q^{mn}}$ over $\Fq$, the exponentiation can be done with $O(m^2n^2 \log mn \log \log mn \log q)$ operations in $\Fq$, with fast multiplication and repeated squaring. Thus the cost of Step~\ref{step3} is $O(km^2n^2 \log mn \log \log mn \log q)$. Let $\alpha$ denote the probability that a given random monic polynomial of degree $mn$ be primitive. Since the number of primitive polynomials of degree $mn$ over $\Fq$ is ${\phi(q^{mn}-1)}/{mn}$, where $\phi$ is Euler's totient function. The value of $\alpha$ is given by $\phi(q^{mn}-1)/(mnq^{mn})$. It is clear that the expected number of times the Algorithm~\ref{Algo1} is iterated to find a primitive MRMM is $1/\alpha$. So the expected number of times Step \ref{step2} to be executed is $1/\alpha$. It is well-known that the probability of a random monic polynomial of degree $mn$ in $\Fq[X]$ being irreducible over $\Fq$ is close to $1/mn$. So the expected number of times Step~\ref{step3} to be executed is $1/(mn \alpha)$. Thus the expected run time of Algorithm~\ref{Algo1} is ${\alpha}^{-1} O(m^2n^2 \log mn \log\log mn \log mnq) + {(mn \alpha)}^{-1}O(km^2n^2 \log mn \log \log mn \log q)$, which, after simplification, can be seen to be equal to $O({\alpha^{-1} mn \log mn \log\log mn}(mn \log mnq + k \log q))$. In view of the fact that for a given number $N$, the number of distinct prime factors of $N$ is asymptotically $\log \log N$ \cite[p. 51]{Hardy}, we can simply omit the second term ``$k\log q$'' inside the big Oh notation. As a consequence, the expected run time of Algorithm~\ref{Algo1} is $\displaystyle O(\alpha^{-1} m^2n^2 \log mn \log\log mn \log mnq)$. Further, by using the well-known lower bound on Euler's totient function due to Landau \cite[Theorem~3.4.2]{JCL} (see also \cite[Fact~2.102]{HAC}), it follows that the expected run time of Algorithm~\ref{Algo1} is given by $\displaystyle O(m^3n^3 \log \log q^{mn}\log mn \log\log mn \log mnq)$.

\begin{remark}
 It is clear that our algorithm finds an efficient primitive MRMM of order $n$ over $\Fqm$ only for small values of $mn$ for which the factorization of $q^{mn}-1$ is known. For large values of $mn$, it would be computationally infeasible to generate all primitive polynomials of degree $mn$ over $\Fq$ due to the rapid growth of the Euler's totient function. For instance, the number of primitive polynomials of degree $100$ over $\mathbb F_2$ is already $5.70767634 \times 10^{27}$. 
Even at the conservative estimate of three bytes per polynomial, this would exceed the total amount of data stored digitally in the world, which was estimated to be 264 exabytes in 2007 \cite{HL}. On the other hand, primitive polynomials of large degree over $\mathbb F_2$ are known; see, for example, some recent papers due to Brent and Zimmerman \cite{BZ1,BZ2,BZ3}. Thus if one intends to use above algorithm to generate primitive MRMM corresponding to all primitive polynomials of degree $mn$ (for small values of $mn$), then it is not an efficient way to do so. In fact, 
there are 
faster algorithms; see, for example, an algorithm due to Porto, Guida and Montolivo \cite{PGM}, which generate all primitive polynomials of degree $D$ given a single primitive polynomial of degree $D$.
\end{remark}

\begin{example}
Let us consider the same polynomial $f(X)$ as given in Example \ref{exam1}. One can verify that $f(X)$ is a primitive over $\mathbb F_2$. The $3$-Horner's matrix of $f(X)$ is given by
\begin{center}
\[H_4(3,f)=\left(%
\begin{array}{cccc}
  X^3 & 0 & 0 & (1+X^2) \\
  1 & X^3 & 0 & X^2 \\
  0 & 1 & X^3 & X\\
  0 & 0 & 1 & (X^3+X^2+X)
\end{array}%
\right).\]
\end{center}
We can express $H_4(3,f)$ in the form of a following matrix polynomial 
$$H_4(3,f)=I_3X^3+C_2X^2+C_1x+C_0,$$ where 
$ C_0=\begin{pmatrix}
  0 & 0 & 0 & 1 \\
  1 & 0 & 0 & 0 \\
  0 & 1 & 0 & 0 \\
  0 & 0 & 1 & 0 
\end{pmatrix}$, 
$C_1=\begin{pmatrix}
  0 & 0 & 0 & 0 \\
  0 & 0 & 0 & 0 \\
  0 & 0 & 0 & 1 \\
  0 & 0 & 0 & 1 
\end{pmatrix}$, and $
C_2= \begin{pmatrix}
  0 & 0 & 0 & 1 \\
  0 & 0 & 0 & 1 \\
  0 & 0 & 0 & 0 \\
  0 & 0 & 0 & 1 
\end{pmatrix}
$.
It is clear that $C_0 \in \GL_4(\mathbb F_2)$. If $\widetilde{T}$ denotes the block companion matrix corresponding to the matrix polynomial, that is, 
$ \widetilde{T} =\begin{pmatrix}
\mathbf{0}&\mathbf{0} &C_0\\
I_3 & \mathbf{0} & C_1\\
\mathbf{0} & I_3 & C_2
\end{pmatrix},
$
then $\widetilde{T} \in \GL_{12}(\mathbb F_2)$ and $o(\widetilde{T})=2^{12}-1$. Moreover, $\Psi_P(\widetilde{T})=f(X)$.
Now corresponding to these $C_0, C_1,$ and $C_2$, we can associate a primitive MRMM of order $3$ over $\mathbb F_{2^4}$.\end{example}
\section{Efficient Implementation}\label{implementation}
In this section, 
we shall restrict ourselves to only binary fields and their extensions. However, all the results can be emulated over an arbitrary finite field. As pointed out in Section \ref{construction}, the efficiency of MRMM constructed through Algorithm~\ref{Algo1} is due to the special structure of the matrices $C_0, C_1, \dots, C_{n-1}$. 
It was also noted in Section~\ref{construction} that the first $(m-1)$ columns of matrix $C_j$ $(1\leq j\leq n-1)$ are zero. Moreover, the matrix $C_0$ has a special 
structure. In fact, it is easy to see that $C_0=R+\widehat{C_0}$, where $R$ is right shift operator given by the matrix

$$R=\begin{pmatrix}
 0 & 0 & \cdots & 0 & 0 \\
  1 & 0 & \cdots & 0 & 0 \\
  0 & 1 & \cdots & 0 & 0 \\
  \vdots & \vdots & \ddots & \vdots & \vdots \\
  0 & 0 & \cdots & 1 & 0
\end{pmatrix}_{ m\times m}
 ,$$ 
 and $\widehat{C_0}$ has all its columns zero except the $m^{\rm th}$ column, which is essentially the last column of $C_0$. The structure of $\widehat{C_0}$ is exactly same as $C_j$, $j\geq 1$.
 
 The following lemma makes implementation of MRMM fast and efficient. 
 \begin{lemma}\label{thm:efficient_lemma}
 For any matrix $A \in M_m(\mathbb{F}_2)$ having all the columns zero except the $m^{\rm th}$ 
column and for any vector ${\mathbf s} = [s_0, s_1, \ldots, s_{m-1}]^{\rm tr} \in \mathbb{F}_2^m$, we have
\[
  A {\mathbf s} =  s_{m-1} {\mathbf v}_m
\]
where ${\mathbf v}_m$ represents the $m^{\rm th}$ column of the matrix $A$.
\end{lemma}
\begin{proof}
Proof is obvious.
\end{proof}
By invoking Lemma \ref{thm:efficient_lemma}, the recurrence relation \eqref{sigmalfsr} can be written as follows:
\begin{eqnarray}\label{eq:efficient_form}
        {\mathbf s}_{i+n} & =&  R{\mathbf s}_i + \widehat{C_0}{\mathbf s}_i + C_1{\mathbf s}_{i+1} + \cdots + C_{n-1}{\mathbf s}_{i+n-1}\nonumber \\
         &=& R{\mathbf s}_i + s_0 {{\mathbf v}{}_m^{0}} + s_1 {\mathbf v}_m^1 + \cdots + s_{n-1} {\mathbf v}_{m}^{n-1}, \label{eqeffi}
\end{eqnarray}
where $s_i$ is the least significant bit (LSB) of ${\mathbf s}_i$, ${\mathbf v}_m^i$ is the $m^{\rm th}$ 
column of the matrix $C_i (0\leq i\leq n-1)$. 

It is clear that \eqref{eqeffi} can be computed by using only one right shift operation and 
at most $n$ bitwise XOR operations instead of matrix multiplications and thus, provides an efficient software realization. 
\section{Tweaked Multiple-Recursive Matrix Method}
\label{nmrmm}
As we know that in \cite{HPW}, a tweak based on Langford arrangement was introduced for the sequences generated by TSRs. In this section, however, we shall consider a slightly different tweak, but based on Langford arrangement itself for the sequences generated by the multiple-recursive matrix method along the similar lines. 

We recall the definition of Langford arrangement \cite{Langford} of a sequence of numbers, which is an important object of study in combinatorics.
 \begin{definition} 
Arrange the numbers $11223344 \cdots gg$ in a sequence such that 
between equal numbers $h$ there are exactly $h$ other numbers. 
This type of arrangement of numbers is known as a \emph{Langford 
arrangement}.
\end{definition}
\begin{example}
For $g=4$ and $g=8$, the Langford arrangements are $41312432$ 
and $6751814657342832$, respectively. 
\end{example}
We define the notion of tweaked primitive multiple-recursive matrix method based on Langford arrangement as follows. 
\begin{definition}
Let $\mathbf{s}_i=
\left(s_i^{(1)},\dots,s_{i}^{(m)}\right) \in \Fq^m \simeq \Fqm$, 
$i=0,1,\dots,$ be the sequence over $\Fqm$ generated by a primitive 
MRMM of order $2g$, where $g$ is a positive integer. Suppose there exists a Langford arrangement for 
the number $g$, and let $\ell_k$ and $r_k$, respectively, denote 
the left and right positions of the number $k$ in the Langford 
arrangement of $g$ from the left. Then $r_k=\ell_k+k+1$. We 
define a sequence $\mathbf{t}^{\infty}=\mathbf{t}_0, \mathbf{t}_1, \dots$ over $\Fqm$ 
obtained from $(\mathbf{s}_i)_{i=0}^{\infty}$ by the following 
recurrence relation:
\begin{eqnarray}
 \mathbf{t}_i=\displaystyle \sum_{j=0}^{i} \mathbf{u}_j \quad \mbox{for} \quad i=0,1, \dots \label{nltsr}\\
 \mbox{where} \quad \mathbf{u}_j=\displaystyle \sum_{k=1}^{g} \mathbf{s}_{2g+j-\ell_k}\star \mathbf{s}_{2g+j-r_k}. \label{nltsr1}
\end{eqnarray}
The operation $\star$ denotes the component-wise multiplication of the 
vectors defined as 
$\mathbf{s}_{2g+i-\ell_k} \star \mathbf{s}_{2g+i-r_k} 
   = \left(s_{2g+i-\ell_k}^{(1)}s_{2g+i-r_k}^{(1)}, \dots, 
     s_{2g+i-\ell_k}^{(m)}s_{2g+i-r_k}^{(m)}\right)$ and $\sum$ denotes the component-wise addition of the vectors.

The system \eqref {nltsr} is called the \emph{tweaked primitive MRMM based on a Langford arrangement
} 
of order $2g$ over $\mathbb F_{q^m}$, while the sequence 
$(\mathbf{t}_i)_{i=0}^{\infty}$ is referred to as the \emph{sequence generated by the tweaked primitive MRMM based on a Langford arrangement}.
\end{definition}
\begin{example}
We consider the Langford arrangement for the number $g=4$ given by $41312432$. In this case, the values of $\ell_i$'s and $r_i$'s are: $\ell_1=2, r_1=4, \ell_2=5, r_2=8, \ell_3=3, r_3=7, \ell_4=1, r_4=6.$

Let $(\mathbf{s}_i)_{i=0}^{\infty}$ be a sequence generated by a primitive MRMM of length $8$ over $\mathbb{F}_{2^m}$. Then the sequence$(\mathbf{t}_i)_{i=0}^{\infty}$ generated by a tweaked primitive 
MRMM based on the above Langford arrangement of length $8$ over $\mathbb{F}_{2^m}$ is given by
$$\mathbf{t}_0=\mathbf{u}_0, \mathbf{t}_1=\mathbf{u}_0 + \mathbf{u}_1, \mbox{and so on},$$ where
$\mathbf{u}_0=\mathbf{s}_{6}\star \mathbf{s}_{4}+\mathbf{s}_{3} \star \mathbf{s}_{0}+\mathbf{s}_{5}\star\mathbf{s}_{1}
 +\mathbf{s}_{7}\star\mathbf{s}_{2},$ 
$\mathbf{u}_1=\mathbf{s}_{7}\star\mathbf{s}_{5}+\mathbf{s}_{4}\star\mathbf{s}_{1}+\mathbf{s}_{6}\star\mathbf{s}_{2}
 +\mathbf{s}_{8}\star \mathbf{s}_{3}$, and so on. 
\end{example}
The following theorem gives the component-wise 
linear complexity of the auxiliary sequence $(\mathbf{u}_i)_{i=0}^{\infty}$ as defined in \eqref{nltsr1}.
\begin{theorem}
\label{CWLCofNLTSR}
Let $$ \mathbf{u}_i= \left(u_i^{(1)},\dots,u_{i}^{(m)}\right) 
\in \Fq^m \simeq \Fqm, i=0,1,\dots,$$ be a sequence as defined in \eqref {nltsr1}. Then for each $1\leq j \leq m$, the linear 
complexity of the $j^{\rm th}$ coordinate sequence 
$u_0^{(j)}, u_1^{(j)}, \dots$ is given by $\frac{mn (mn+1)}{2}$.
\end{theorem}
\begin{proof}
For each $1\leq j \leq m$, it follows from \eqref{nltsr1} that 
$$ u_i^{(j)}=\displaystyle \sum_{k=1}^{g}s_{2g+i-\ell_k}^{(j)}s_{2g+i-r_k}^{(j)}, 
  \quad i=0,1,\dots.$$
The Corollary \ref{compLC} ensures that the linear complexity of the 
component sequences $s_i^{(j)}$, $i=0,1, \dots,$ is $mn$. Now $s_i^{(j)}$, $i=0,1, \dots,$ can be thought of as a sequence generated by a primitive LFSR of order $mn$ and thus it follows from \cite[Section III]{KEY} that 
the linear complexity of the sequence 
$({u}_i^{(j)})_{i=0}^{\infty}$ is $\frac{mn (mn+1)}{2}$.   
\end{proof}
In view of Theorem \ref{CWLCofNLTSR}, the component-wise linear complexity of the sequence $(\mathbf{t}_i)_{i=0}^{\infty}$ generated by tweaked primitive MRMM based on Langford arrangement of order $2g$ over $\Fqm$ is 
of the order of $mn(mn+1)/2$, which is $(mn+1)/2$ times more than that of the sequences generated by the usual primitive MRMM. 
 \section*{Acknowledgments}
 We would like to thank anonymous referees for their several insightful comments that have significantly improved the quality of our manuscript and for drawing our attention to \cite{HML,JCL}. We would like to thank Samrith Ram for his careful reading of the initial draft of our manuscript and for his many useful discussions, particularly, in Section \ref{algorithm}. We would also like to thank Shri T. S. Raghavan for his valuable guidance and support.


\end{document}